\documentclass[12pt]{article}
\usepackage{fullpage}
\usepackage{amsmath,amsfonts,amssymb,amsthm}
\usepackage[english]{babel}
\usepackage[utf8]{inputenc}

\newtheorem{theorem}{Theorem}
\newtheorem{lemma}{Lemma}
\newtheorem{corollary}{Corollary}

\newtheorem{fact}{Fact}
\newtheorem{proposition}{Proposition}
\newtheorem{definition}{Definition}
\newtheorem{conjecture}{Conjecture}

\newtheoremstyle{foostyle}
  {\topsep}
  {\topsep}
  {}
  {0pt}
  {\bfseries}
  {.}
  { }
  {\thmname{#1}\thmnumber{ #2}\thmnote{ (#3)}}
\theoremstyle{foostyle}
\newtheorem{example}{Example}

\newcommand{\etal}{\mbox{\textit{et al}.\ }}

\newcommand{\QQ}{\mathbb{Q}}
\newcommand{\ZZ}{\mathbb{Z}}

\newcommand{\CC}{\mathbb{C}}
\newcommand{\TT}{\mathbb{T}}

\newcommand{\II}{\mathbb{I}}

\newcommand{\ii}{\mathtt{i}}

\DeclareMathOperator{\diag}{diag}

\DeclareMathOperator{\SwAut}{\mathtt{SwAut}}
\DeclareMathOperator{\Circ}{\mathtt{Circ}}
\DeclareMathOperator{\Gal}{\mathtt{Gal}}

\DeclareMathOperator{\qmod}{mod}

\newcommand{\bra}[1]{\langle #1 |}
\newcommand{\ket}[1]{| #1 \rangle}
\newcommand{\braket}[2]{\langle #1 | #2 \rangle}
\newcommand{\ketbra}[2]{| #1 \rangle\langle #2 |}

\newcommand{\ignore}[1]{}

\newcommand{\ohf}{\tfrac{1}{2}}
\newcommand{\thf}{\tfrac{3}{2}}
\newcommand{\qdiv}[2]{\lfloor #1 / #2 \rfloor}
\newcommand{\rmod}[2]{#1 \qmod #2}
\newcommand{\mygcd}{\gcd}

\newif\ifnotesw\noteswtrue
   {\ifnotesw\marginpar[\hfill\(\top\)]{\(\top\)}\fi}%
   {\ifnotesw\marginpar[\hfill\(\bot\)]{\(\bot\)}\fi}

\newcommand{\mnote}[1]%
    {\ifnotesw\marginpar%
        [{\scriptsize\begin{minipage}[t]{\marginparwidth}
        \raggedleft#1%
                        \end{minipage}}]%
        {\scriptsize\begin{minipage}[t]{\marginparwidth}
        \raggedright#1%
                        \end{minipage}}%
    \fi}

\title{
Universality in perfect state transfer
}

\author{
Erin Connelly\thanks{Department of Mathematics and Statistics, Haverford College.}
\and
Nathaniel Grammel\thanks{Department of Computer Science, NYU Polytechnic.}
\and
Michael Kraut\thanks{Mathematics Department, University of California at Santa Cruz.}
\and
Luis Serazo\thanks{Mathematics and Statistics Department, Vassar College.}
\and
Christino Tamon\thanks{Department of Computer Science, Clarkson University. Contact: tino@clarkson.edu}
}

\date{\today}

\begin{document}
\maketitle
\bibliographystyle{alpha}

\begin{abstract}
A continuous-time quantum walk on a graph is a matrix-valued function 
$\exp(-\ii At)$ over the reals, where $A$ is the adjacency matrix of the graph.
Such a quantum walk has {universal} perfect state transfer 
if for all vertices $u,v$, there is a time where
the $(v,u)$ entry of the matrix exponential has unit magnitude.
We prove new characterizations of graphs with universal perfect state transfer.
This extends results of Cameron \etal (Linear Algebra and Its Applications, 455:115-142, 2014).
Also, we construct non-circulant families of graphs with universal perfect state transfer. 
All prior known constructions were circulants.
Moreover, we prove that if a circulant, whose order is prime, prime squared, or a power of two, 
has universal perfect state transfer then its underlying graph must be complete.
This is nearly tight since there are universal perfect state transfer circulants
with non-prime-power order where some edges are missing. 

\vspace{.1in}
\end{abstract}


\section{Introduction}

A {\em continuous-time quantum walk} on a graph is given by the one-parameter matrix-valued map 
$U(t) = \exp(-\ii At)$, where $A$ is the Hermitian adjacency matrix of the graph.
This notion was introduced by Farhi and Guttman \cite{fg98} to study quantum algorithms for search problems.
Using this formulation, Bose \cite{b03} studied problems related to information transmission in
quantum spin chains. In such a quantum walk, we say that there is {\em perfect state transfer} 
from vertex $u$ to vertex $v$ at time $\tau$ if the $(v,u)$ entry of $U(\tau)$ 
has unit magnitude.

Kay \cite{k11} showed that if a graph, whose adjacency matrix is real symmetric, has 
perfect state transfer from $u$ to $v$ and also from $u$ to $w$, then $v$ must be equal to $w$. 
Cameron \etal \cite{cfghst14} showed that it is possible to violate this ``monogamy'' property 
in graphs with complex Hermitian adjacency matrices. 
They studied graphs with a {\em universal} property where perfect state transfer occurs
between every pair of vertices.
The smallest nontrivial example is the circulant $\Circ(0,-\ii,\ii)$.
Some work related to universality in state transfer for quantum computing applications
may be found in \cite{chiral-qwalk}.

We may view certain graphs with Hermitian adjacency matrices as {\em gain} graphs. 
These are graphs whose ``directed'' edges are labeled with elements from a group $\Gamma$. 
If the edge $(u,v)$ is labeled with group element $g \in \Gamma$, 
then the reversed edge $(v,u)$ is labeled with $g^{-1}$. 
If $\Gamma$ is the circle group, we get the complex unit gain graphs (see Reff \cite{r12}).
In our case, we simply require that edges in opposite directions have weights which
are complex conjugates to each other.

Our main goal is to characterize graphs with universal perfect state transfer.
Cameron \etal \cite{cfghst14} proved strong necessary conditions for graphs with the universal
perfect state transfer property. They showed that such graphs must have distinct eigenvalues,
their unitary diagonalizing matrices must be type-II (see Chan and Godsil \cite{cg10}),
and their switching automorphism group must be cyclic. 
A spectral characterization for circulants with the universal property was also proved in \cite{cfghst14}.

In this work, we extend some of the observations from Cameron \etal \cite{cfghst14}.
More specifically, we prove new characterizations of graphs with universal perfect state transfer. 
The first characterization exploits the fact that the unitary diagonalizing matrix of the graph 
admits a canonical form. This allows us to show a tight connection between the spectra of the
graph with the perfect state transfer times. Our second characterization is on circulants
with the universal property. It involves the set of minimum times when perfect state transfer
occur between pairs of vertices. We prove that these minimum times are equally spaced on the
periodic time interval when perfect state transfer returns to the start vertex if and only if
the graph is circulant.
This complements the observation in \cite{cfghst14} which characterizes the switching 
automorphism group of circulants with the universal property.

Most of the examples studied in \cite{cfghst14} were circulants whose nonzero weights are $\pm\ii$.
Here, we provide a construction of {\em non-circulant} graphs of composite order with the universal property. 
To the best of our knowledge, this is the first known example of such family of graphs.
We show that these families are non-circulant by appealing to our second characterization 
above (based on spacings of the minimum perfect state transfer times). 

Finally, we provide a nearly tight characterization of circulants with the universal property 
in terms of the number of nonzero coefficients. We show that if a circulant has universal perfect state
tranfer and its order is prime, square of a prime, or a power of two, then all of its off-diagonal
coefficients must be nonzero. As a partial converse, we show an infinite family of circulants
with universal perfect state transfer whose order is not a prime power and where some off-diagonal 
coefficients are zero.

We conclude by studying universal perfect state transfer in complex unit gain graphs. 
The only known examples of complex unit gain graphs with the universal property are 
the circulants $K_{2}$ and $\Circ(0,-\ii,\ii)$. We conjecture that this set is unique.

For a recent survey and a comprehensive treatment of quantum walk on graphs, we refer the
interested reader to Godsil \cite{g12,godsil-book}.


\section{Preliminaries}

A weighted graph $G = (V,E,w)$ is defined by a vertex set $V$, an edge set $E$ 
and a weight function $w: E \rightarrow \CC$.
If $G$ has $n$ vertices, we will often identify the vertex set with $\ZZ/n\ZZ$.
The adjacency matrix $A(G)$ of graph $G$ is a $n \times n$ matrix defined as
$A(G)_{k,j} = w(j,k)$ if $(j,k) \in E$, and $A(G)_{k,j} = 0$ otherwise.
A graph is Hermitian if its adjacency matrix is (see \cite{gm15}).

For complex numbers $a_{0},\ldots,a_{n-1} \in \CC$, we let $C = \Circ(a_{0},\ldots,a_{n-1})$ 
denote the circulant matrix of order $n$ where $C_{jk} = a_{k-j}$, for $j,k \in \ZZ/n\ZZ$.
If $C$ is Hermitian, note that $a_{0}$ must be real and $a_{n-j} = \overline{a}_{j}$,
for $j=1,\ldots,n-1$.
It is known that any circulant is diagonalized by the Fourier matrix $F_{n}$
defined by $\bra{j}F_{n}\ket{k} = \zeta_{n}^{jk}/\sqrt{n}$.
Here, $\zeta_{n} = e^{2\pi\ii/n}$ denotes a primitive $n$th root of unity.

We call a matrix {\em flat} if all of its entries have the same magnitude.
A matrix is type-II if it is flat and unitary (see Chan and Godsil \cite{cg10}).
Note that the Fourier matrix is type-II.
A {\em monomial} matrix is a product of a permutation matrix and an invertible diagonal matrix 
(see Davis \cite{davis}). 
Two matrices $A$ and $B$ are {\em switching equivalent} if $MA = BM$ for some monomial matrix.
The {\em switching automorphism group} of a graph $G$, denoted $\SwAut(G)$,
is the group of all monomial matrices which commute with $A(G)$. 
This generalizes the notion of an automorphism group of a graph.

For more background on algebraic graph theory, see Godsil and Royle \cite{gr01}.


\section{Basic properties}

The following result shows strong necessary conditions for graphs with universal perfect state transfer. 

\begin{theorem} (Cameron \etal \cite{cfghst14}) \\
Let $G$ be a Hermitian graph with universal perfect state transfer. Then, the following hold:
\begin{enumerate}
\item All eigenvalues of $G$ are distinct.
\item The adjacency matrix of $G$ is unitarily diagonalized by a flat matrix.
\item The switching automorphism group of $G$ is cyclic whose order divides the size of $G$.
\end{enumerate}
\end{theorem}

\par\noindent
We show some additional properties of graphs with universal perfect state transfer.

\begin{definition}
Let $G$ be a graph with universal perfect state transfer.
For every pair of vertices $v$ and $w$ of $G$, we let $T_{v,w}$ denote the set of times where
perfect state transfer occurs from $v$ to $w$. That is,
\begin{equation}
T_{v,w} := \{t \in \mathbb{R}^{+} : |\bra{w}e^{-\ii A(G)t}\ket{v}| = 1\}.
\end{equation}
\end{definition}

\begin{fact} 
Let $G$ be a graph with universal perfect state transfer.
For each pair of vertices $v$ and $w$, $T_{v,w}$ is a discrete additive subgroup of $\mathbb{R}$.

\begin{proof}
See Godsil \cite{godsil-book} or Cameron \etal \cite{cfghst14}.
\end{proof}
\end{fact}

Since $T_{v,w}$ is a discrete additive subgroup of the reals,
it has a smallest element. We will denote the minimum element of the above set as 
$t_{v,w} := \min T_{v,w}$.

\begin{lemma} \label{lem:hit-and-return}
Let $G$ be a graph with universal perfect state transfer
and let $u$ be a vertex of $G$.
Then, for all vertices $v \neq u$, we have $t_{u,v} < t_{u,u}$.

\begin{proof}
If $t_{u,v} > t_{u,u}$, let $q$ be the largest integer for which $q t_{u,u} < t_{u,v}$.
Then, $\hat{t} = t_{u,v} - q t_{u,u}$ is an element of $T_{u,v}$ which is smaller than $t_{u,v}$.
\end{proof}
\end{lemma}

\begin{lemma} \label{lem:source-condition}
Let $G$ be a graph and let $u$ be an arbitrary vertex of $G$.
Then $G$ has universal perfect state transfer 
if and only if
perfect state transfer occurs from $u$ to all vertices of $G$.

\begin{proof}
It suffices to prove only one direction since the other direction is immediate.
Suppose that perfect state transfer occurs from $u$ to all vertices.
By Lemma \ref{lem:hit-and-return}, the quantum walk starting at $u$ {\em visits} 
all the other vertices before returning to $u$.
If $t_{u,v} < t_{u,w}$ then perfect state transfer occurs from $v$ to $w$.
But, there is also perfect state transfer from $w$ to $v$ since the quantum walk
has perfect state transfer from $w$ back to $u$ (at time $t_{u,u} - t_{u,w}$)
and then from $u$ to $v$ (at time $t_{u,v}$). 
This proves that there is perfect state transfer between every pair of vertices.
\end{proof}
\end{lemma}


\section{Canonical flatness}

A flat unitary matrix is also called a {\em type-II} matrix (see \cite{cg10}).
We say that a type-II matrix is in {\em canonical} form if both its first row and its first column
are the all-one vector.

\begin{lemma} \label{lemma:canonical-flatness}
Let $G$ be a Hermitian graph on $n$ vertices with universal perfect state tranfer. 
Then $G$ is unitarily diagonalized by a type-II matrix $X$ in canonical form: 
\begin{equation} \label{eqn:canonical-flat}
X =
\frac{1}{\sqrt{n}}
\begin{bmatrix}
1 & 1 & 1 & \ldots & 1 \\
1 & e^{\ii\alpha_{1,1}} & e^{\ii\alpha_{1,2}} & \ldots & e^{\ii\alpha_{1,n-1}} \\
1 & e^{\ii\alpha_{2,1}} & e^{\ii\alpha_{2,2}} & \ldots & e^{\ii\alpha_{2,n-1}} \\
\vdots & \vdots & \vdots & & \vdots \\
1 & e^{\ii\alpha_{n-1,1}} & e^{\ii\alpha_{n-1,2}} & \ldots & e^{\ii\alpha_{n-1,n-1}} 
\end{bmatrix}
\end{equation}
where $\alpha_{j,k} \in [0,2\pi)$. 

\begin{proof}
Suppose $A$ is the Hermitian adjacency matrix of $G$ that is unitarily diagonalized by $Z$. 
Then, $AZ = Z\Lambda$ where $\Lambda$ is a diagonal matrix of the eigenvalues of $A$.
The columns of $Z$ are the eigenvectors of $A$ which we will denote as
$\ket{z_{0}},\ldots,\ket{z_{n-1}}$. 
Let $D$ be a diagonal matrix defined as $D_{jj} = 1/\braket{0}{z_{j}}$.
Then, $\tilde{Z} = ZD$ is also a flat unitary which diagonalizes $A$ but
with $\braket{0}{\tilde{z}_{k}} = 1$, for each $k=0,\ldots,n-1$. 
Next, consider a diagonal switching matrix $S$ defined as
$S_{jj} = 1/\braket{j}{\tilde{z}_{0}}$.
Then, we have
\begin{equation}
\tilde{A} = (S\tilde{Z})\Lambda(S\tilde{Z})^{-1}.
\end{equation}
Note $\tilde{A} = SAS^{-1}$ is switching equivalent to $A$.
Since $X = S\tilde{Z}$ is a flat unitary matrix of the claimed form, we are done.
\end{proof}
\end{lemma}

\medskip

\begin{corollary} \label{cor:flatness}
Let $M$ be a flat unitary matrix in canonical form.
Then, except for the first row and the first column,
the row sums and the column sums of $M$ are zero.

\begin{proof}
Since the columns are orthonormal and the first column is the all-one vector,
it is clear that the column sums must be zero. 
The row sums are zero since $M^{T}$ is unitary whenever $M$ is.
\end{proof}
\end{corollary}

\medskip

Using Lemma \ref{lemma:canonical-flatness}, we show a spectral characterization of graphs with
universal perfect state transfer.

\begin{theorem} \label{thm:canonical-flatness-upst}
Let $G$ be a $n$-vertex Hermitian graph with eigenvalues $\lambda_{0},\ldots,\lambda_{n-1}$.
Suppose $G$ is diagonalized by a canonical type-II matrix $X$, 
where $X_{j,k} = e^{\ii\alpha_{j,k}}/\sqrt{n}$ with $\alpha_{j,k} \in [0,2\pi)$ 
and $\alpha_{j,k} = 0$ if either $j$ or $k$ is zero.
Then, 
$G$ has universal perfect state transfer
if and only if 
for each $\ell = 0,\ldots,n-1$, 
there is $t_{\ell} \in \mathbb{R}$ so that
for all $k = 1,\ldots,n-1$, we have
\begin{equation} \label{eqn:pst-condition}
(\lambda_{k} - \lambda_{0})t_{\ell} = \alpha_{\ell,k}.
\end{equation}

\begin{proof}
Let $A$ be the Hermitian adjacency matrix of $G$.
We denote the $k$th column of $X$ as $\ket{\lambda_{k}}$ which is the eigenvector
of $A$ corresponding to eigenvalue $\lambda_{k}$. Thus, 
\begin{equation}
e^{-\ii At} = \sum_{k=0}^{n-1} e^{-\ii\lambda_{k}t} \ketbra{\lambda_{k}}{\lambda_{k}}.
\end{equation}

\par\noindent ($\Rightarrow$)
Assume $G$ has universal perfect state transfer. Suppose that perfect state transfer from
vertex $0$ to vertex $\ell$ occurs at time $t_{\ell} \in \mathbb{R}$ with phase $e^{\ii\theta_{\ell}}$. 
Then,
\begin{equation}
\bra{\ell}e^{-\ii At_{\ell}}\ket{0} 
	= \sum_{k=0}^{n-1} e^{-\ii\lambda_{k}t_{\ell}} \braket{\ell}{\lambda_{k}}\braket{\lambda_{k}}{0} \\
	= \frac{1}{n} \sum_{k=0}^{n-1} e^{-\ii\lambda_{k}t_{\ell}} e^{\ii\alpha_{\ell,k}}. 
\end{equation}
Moreover, we have
\begin{equation}
e^{\ii\theta_{\ell}}
	= \frac{1}{n} \left[ e^{-\ii\lambda_{0}t_{\ell}} 
		+ \sum_{k=1}^{n-1} e^{-\ii(\lambda_{k}t_{\ell} - \alpha_{\ell,k})} \right].
\end{equation}
So, for each $k=1,\ldots,n-1$, we have
$\lambda_{k}t_{\ell} - \alpha_{\ell,k} = \lambda_{0}t_{\ell}$,
which shows these conditions are necessary for universal perfect state transfer.

\medskip

\par\noindent ($\Leftarrow$)
Suppose that for each $\ell$, there is a time $t_{\ell}$ so that for each $k \neq 0$, 
\begin{equation} 
(\lambda_{k} - \lambda_{0})t_{\ell} = \alpha_{\ell,k}.
\end{equation}
Then, 
\begin{equation}
\bra{\ell}e^{-\ii At_{\ell}}\ket{0} 
	= \frac{1}{n} \sum_{k=0}^{n-1} e^{-\ii\lambda_{k}t_{\ell}} e^{\ii\alpha_{\ell,k}}
	= e^{-\ii\lambda_{0}t_{\ell}}.
\end{equation}
This shows that there is perfect state transfer from $0$ to $\ell$.
Therefore, there is perfect state transfer from $0$ to all vertices.
By Lemma \ref{lem:source-condition}, this shows there is perfect state transfer between
every pair of vertices.
\end{proof}
\end{theorem}


\section{Circulants revisited}

Cameron \etal \cite{cfghst14} showed the following result on the switching automorphism group of
circulants with universal perfect state transfer.

\begin{theorem} \label{thm:circulant-cyclic}
(Cameron \etal \cite{cfghst14}) \\
Let $G$ be a graph with universal perfect state transfer.
Then, $G$ is switching isomorphic to a circulant
if and only if
$\SwAut(G)$ is cyclic of order $n$.
\end{theorem}

\par\noindent
In what follows, we provide new characterizations of circulants with universal perfect state transfer.
The first one is based on the set of times when perfect state transfer occur. 
The second one is based on the explicit form of allowable weights.

\medskip
Recall that $T_{u,v}$ is the set of times (positive real numbers) when perfect state transfer
occur from vertex $u$ to vertex $v$. Also, we denote $t_{u,v}$ as the smallest element of $T_{u,v}$.

\begin{theorem} \label{thm:circulant-timing}
Let $G$ be a $n$-vertex graph with universal perfect state transfer.
Assume that $t_{0,k} < t_{0,k+1}$ for all $k=1,\ldots,n-2$ and that
$t_{0,1} = \min\{t_{k,k+1} : k \in \ZZ/n\ZZ\}$.
Then, $G$ is switching isomorphic to a circulant
if and only if
\begin{equation}
t_{k,k+1} = t_{0,1},
\end{equation} 
for all $k \in \ZZ/n\ZZ$.

\begin{proof}
Let $A$ be the adjacency matrix of $G$.

($\Rightarrow$) 
Suppose that $G$ is switching isomorphic to a circulant.
Consider the set of times when perfect state transfer occur in $G$:
\begin{equation}
T = \{t \in \mathbb{R}^{+} : \exists j,k \in \ZZ/n\ZZ, \ |\bra{k}e^{-\ii At}\ket{j}| = 1\}.
\end{equation}
Since $T$ is a discrete additive subgroup of $\mathbb{R}$, it has a minimum.
Without loss of generality, assume that $t_{0,1} = \min T$.
Thus,
\begin{equation}
\bra{1}e^{-\ii A t_{0,1}}\ket{0} = \gamma,
\end{equation}
for some $\gamma \in \mathbb{T}$.
By Theorem \ref{thm:circulant-cyclic}, $\SwAut(G)$ is cyclic of order $n$.
We may assume that $\tilde{P} = P_{\phi}D$ generates $\SwAut(G)$ where
$\phi = (0 \ 1 \ \ldots \ n-1)$ and $D$ is a diagonal switching matrix.
Since $\tilde{P}$ is a switching automorphism, $\tilde{P}^{-1}A\tilde{P} = A$,
which implies $e^{-\ii At} = \tilde{P}^{-1}e^{-\ii At}\tilde{P}$.
Therefore,
\begin{equation}
\bra{2}e^{-\ii A t_{0,1}}\ket{1} = 
\tilde{\gamma}\bra{1}\tilde{P}^{-1}e^{-\ii A t_{0,1}}\tilde{P}\ket{0} =
\tilde{\gamma}\bra{1}e^{-\ii A t_{0,1}}\ket{0},
\end{equation}
for some $\tilde{\gamma} \in \mathbb{T}$.
So, perfect state transfer occurs from $1$ to $2$ at time $t_{0,1}$.
By repeatedly using the same argument, we see that perfect state transfer occurs 
from $k$ to $k+1$ at time $t_{0,1}$, for $k=0,1,\ldots,n-2$.

($\Leftarrow$) 
Suppose $t_{k,k+1} = t_{0,1}$, for $k \in \ZZ/n\ZZ$.
This implies that at time $t_{0,1}$ perfect state transfer occurs from $k$ to $k+1$
(simultaneously) from $k$ to $k+1$ for each $k \in \ZZ/n\ZZ$. 
So, assume 
\begin{equation}
\bra{k+1}e^{-\ii At_{0,1}}\ket{k} = \gamma_{k},
\end{equation} 
for some complex unit weight $\gamma_{k} \in \mathbb{T}$.
Thus,
\begin{equation}
e^{-\ii At_{0,1}} =
\begin{bmatrix}
0 & 0 & \ldots & 0 & \gamma_{n-1} \\
\gamma_{0} & 0 & \ldots & 0 & 0 \\
0 & \gamma_{1} & \ldots & 0 & 0 \\
\vdots & \vdots & \ldots & \vdots &  \vdots \\
0 & 0 & \ldots & \gamma_{n-2} & 0
\end{bmatrix}
=
\begin{bmatrix}
0 & 0 & \ldots & 0 & 1 \\
1 & 0 & \ldots & 0 & 0 \\
0 & 1 & \ldots & 0 & 0 \\
\vdots & \vdots & \ldots & \vdots &  \vdots \\
0 & 0 & \ldots & 1 & 0
\end{bmatrix}
\begin{bmatrix}
\gamma_{0} & 0 & \ldots & 0 \\
0 & \gamma_{1} & \ldots & 0 \\
0 & 0 & \ldots & 0 \\
\vdots & \vdots & \ldots & \vdots \\
0 & 0 & \ldots & \gamma_{n-1} 
\end{bmatrix},
\end{equation}
which shows that $e^{-\ii At_{0,1}}$ is a monomial matrix.
Since it commutes with $A$, it is a switching automorphism of $G$.
Moreover, it generates $\SwAut(G)$.
Thus, by Theorem \ref{thm:circulant-cyclic}, $G$ is switching isomorphic to a circulant.
\end{proof}
\end{theorem}

\bigskip

Next we find explicit forms for the weights on circulants which have universal perfect state transfer.
But, first we state a spectral characterization of circulants with universal perfect state transfer
proved by Cameron \etal \cite{cfghst14}.

\begin{theorem} \label{thm:circulant-upst-eigenvalue}
(Cameron \etal \cite{cfghst14}) \\
Let $G$ be a graph that is switching equivalent to a circulant. Then
$G$ has universal perfect state transfer 
if and only if 
for some integer $q$ coprime with $n$,
for real numbers $\alpha, \beta$ with $\beta > 0$, 
the eigenvalues of $G$ are given by 
\begin{equation} \label{eqn:eigenvalue-form}
\lambda_{k} = \alpha + \beta(qk + c_{k}n),
\ \hspace{.5in} \
k =0,1,\ldots,n-1,
\end{equation} 
where $c_{k}$ are integers.
\end{theorem}

Let $A$ be the adjacency matrix of $G$.
In Theorem \ref{thm:circulant-upst-eigenvalue},
we may assume $\alpha = 0$ by allowing a diagonal shift $A + \alpha\II$,
which does not affect the quantum walk.
Furthermore, we may assume $\beta = 1$ by allowing the time scaling $\frac{1}{\beta}A$,
which does not affect perfect state transfer.
Finally, we may multiply the adjacency matrix with the multiplicative inverse of $q$ modulo $n$
(to cancel the factor $q$ in $qk$).
In summary, we have the following.

\begin{corollary} \label{cor:simpler-eigenvalue-form}
Let $G$ be a graph that is switching equivalent to a circulant. Then
$G$ has universal perfect state transfer 
if and only if 
$G$ is switching equivalent to a graph with eigenvalues 
\begin{equation} \label{eqn:simpler-form}
\lambda_{k} = k + c_{k}n,
\ \hspace{.5in} \
k = 0,1,\ldots,n-1
\end{equation}
where $c_{k}$ are integers.
\end{corollary}

Next, we show a general form for the coefficients of a circulant which has universal 
perfect state transfer.

\begin{theorem} \label{thm:circ-coeff-upst}
Let $\Circ(a_{0},\ldots,a_{n-1})$ be a circulant with universal perfect state transfer. 
Then, 
we have
\begin{equation}
a_{j} = \frac{1}{\zeta_{n}^{-j}-1} + \sum_{k=0}^{n-1} c_{k}\zeta_{n}^{-jk},
\hspace{0.2in}
j=1,\ldots,n-1 
\end{equation}
for integers $c_{k}$, where $k=0,\ldots,n-1$.

\begin{proof}
Let $G = \Circ(a_{0},\ldots,a_{n-1})$.
By Corollary \ref{cor:simpler-eigenvalue-form}, 
the eigenvalues of $G$ are of the form $\lambda_{k} = k + c_{k}n$,
for some integers $c_{k}$, where $k \in \ZZ/n\ZZ$.
Since circulants are diagonalized by the Fourier matrix (see Biggs \cite{biggs}),
the coefficients of $G$ are given by
\begin{equation}
a_{j} = \frac{1}{n} \sum_{k=0}^{n-1} \lambda_{k}\zeta_{n}^{-jk},
\ \hspace{.5in} \
j=1,\ldots,n-1.
\end{equation}
Using the assumed form of $\lambda_{k}$, we get
\begin{equation}
a_{j} 
	= \frac{1}{n} \sum_{k=0}^{n-1} (k + c_{k}n)\zeta_{n}^{-jk} 
	= \frac{1}{n} \sum_{k=0}^{n-1} k\zeta_{n}^{-jk} + \sum_{k=0}^{n-1} c_{k}\zeta_{n}^{-jk}.
\end{equation}
Let $U = \sum_{k=0}^{n-1} k\zeta_{n}^{-jk}$ and
$L = \sum_{k=1}^{n-2} \sum_{\ell=1}^{k} \zeta_{n}^{-j\ell}$.
Note that 
\begin{equation}
L + U = (n-1)\sum_{k=1}^{n-1} \zeta_{n}^{-jk} = 1-n.
\end{equation}
Now, we have
\begin{eqnarray}
L & = & \sum_{k=1}^{n-2} \sum_{\ell=1}^{k} \zeta_{n}^{-j\ell} 
	= \sum_{k=1}^{n-2} \left(\frac{\zeta_{n}^{-j(k+1)}}{\zeta_{n}^{-j}-1} - 1\right) \\
	& = & \frac{1}{\zeta_{n}^{-j}-1}
			\left( \sum_{k=2}^{n-1} \zeta_{n}^{-jk} - (n-2)\zeta_{n}^{-j} \right) \\
	& = & 1-n - \frac{n}{\zeta_{n}^{-j}-1}.
\end{eqnarray}
Thus, $U = n/(\zeta_{n}^{-j}-1)$. 
Therefore, 
\begin{equation}
a_{j} = \frac{1}{\zeta_{n}^{-j}-1} + \sum_{k=0}^{n-1} c_{k}\zeta_{n}^{-jk}.
\end{equation}
\end{proof}
\end{theorem}


\section{Non-circulants with universal state transfer}

In this section, we show a construction of a family of non-circulant graphs with universal 
perfect state transfer. 
This provides the first known examples of non-circulant families with universal perfect state transfer.

\newcommand{\Idx}{\vartheta}
\newcommand{\Row}[1]{\Idx_{#1}}
\newcommand{\Col}[1]{\Idx_{#1}}

\begin{theorem} \label{thm:noncirculant-upst-graph}
Let $n = ab$ be an integer where $a \ge b \ge 2$ are integers.
Fix an integer $\beta \ge 2$ and for a positive integer $d$, 
let $\vartheta_{d}$ be a function which maps elements of $\ZZ/n\ZZ$ 
to the positive integers defined as
\begin{equation}
\Idx_{d}(x) := \beta\qdiv{x}{d}d + \rmod{x}{d}.
\end{equation}
Let $X$ be a $n \times n$ matrix whose $(j,k)$-entry is given by
\begin{equation}
X_{j,k} 
= \frac{1}{\sqrt{n}}
	\zeta_{\beta n}^{\Row{a}(j)\Col{b}(k)}.
\end{equation}
Then, $X$ is type-II. 
Moreover, if $G$ is the graph with eigenvalues $\{\Col{b}(k) : k=0,1,\ldots,n-1\}$
whose adjacency matrix is unitarily diagonalized by $X$,
then $G$ has universal perfect state transfer.

\begin{proof}
First, we show that $X$ is type-II. It is clear that $X$ is flat from its definition.
So, it suffices to show that the columns of $X$ form an orthonormal set. 
In what follows, for $k,\ell \in \ZZ/n\ZZ$, let
$Q_{b}(k,\ell) = \qdiv{k}{b} - \qdiv{\ell}{b}$
and
$R_{b}(k,\ell) = \rmod{k}{b} - \rmod{\ell}{b}$.
If $X_{k}$ and $X_{\ell}$ are the $k$th and $\ell$th columns of $X$, then
\begin{eqnarray}
\braket{X_{\ell}}{X_{k}}
	& = & \frac{1}{n} \sum_{j=0}^{n-1} \zeta_{\beta n}^{M_{b}(k,\ell)\Row{a}(j)},
	\ \ 
	\mbox{where $M_{b}(k,\ell) = \beta Q_{b}(k,\ell)b + R_{b}(k,\ell)$} \\
	& = & \frac{1}{n} \sum_{r=0}^{a-1} \zeta_{\beta n}^{M_{b}(k,\ell)r}
			\sum_{q=0}^{b-1} \left(\zeta_{n}^{aM_{b}(k,\ell)} \right)^{q}.
\end{eqnarray}
But, note that for any integer $M \neq 0$, provided $\zeta_{b}^{M} \neq 1$, we have
\begin{equation}
\sum_{q=0}^{b-1} (\zeta_{n}^{aM})^{q} 
	= \frac{\zeta_{n}^{abM} - 1}{\zeta_{n}^{aM}-1} 
	= 0,
\end{equation}
since $\zeta_{n}^{ab} = 1$.
So, if $R_{b}(k,\ell) \neq 0$, then $M_{b}(k,\ell) \neq 0\pmod{b}$, which implies
$\braket{X_{\ell}}{X_{k}} = 0$.
On the other hand, if $R_{b}(k,\ell) = 0$, then $Q_{b}(k,\ell) \neq 0$, 
for otherwise $k = \ell$. Here, we have
\begin{equation}
\braket{X_{\ell}}{X_{k}}
	= \frac{1}{n} 
			\sum_{q=0}^{b-1} \zeta_{1}^{q\beta Q_{b}(k,\ell)} 
			\sum_{r=0}^{a-1} \left(\zeta_{a}^{Q_{b}(k,\ell)}\right)^{r}
\end{equation}
Thus, $\braket{X_{\ell}}{X_{k}} = 0$ also holds.

Next, we show that $G$ has universal perfect state transfer. 
If we let 
\begin{equation}
\alpha_{j,k} = \frac{2\pi}{\beta n} \Row{a}(j)\Col{b}(k),
\end{equation} 
for $j,k \in \ZZ/n\ZZ$, 
then $X_{j,k} = e^{\ii\alpha_{j,k}}$.
For each $j \in \ZZ/n\ZZ$, let $t_{j} = (2\pi/\beta n)\Row{a}(j)$. 
Then,
\begin{equation}
\lambda_{k} t_{j} = \alpha_{j,k}
\end{equation}
holds for all $k \in \ZZ/n\ZZ$, since $\lambda_{k} = \Col{b}(k)$ and $\lambda_{0} = 0$.
By Theorem \ref{thm:canonical-flatness-upst}, this shows that $G$ has universal perfect state transfer.

Finally, we show that $G$ is not switching equivalent to a circulant.
By Theorem \ref{thm:circulant-timing}, it suffices to show that $t_{j+1}-t_{j}$ 
are not all equal. From the definition of $\Row{a}(j)$, we have
$t_{1} - t_{0} = \tfrac{2\pi}{\beta n}$ whereas $t_{a} - t_{a-1} = \tfrac{2\pi}{\beta n}((\beta - 1)a + 1)$.
\end{proof}
\end{theorem}

\vspace{.1in}

\begin{example}
For even $k \ge 2$, let $G_{k}$ be a graph of order $4$ with eigenvalues $\{0,1,k,k+1\}$.
Let $X$ be the following unitary matrix (whose columns are the eigenvectors of $G_{k}$):
\begin{equation}
X = 
\begin{bmatrix}
1 & 1 & 1 & 1 \\
1 & e^{\ii\pi/k} & e^{\ii\pi} & e^{\ii\pi(k+1)/k} \\
1 & e^{\ii\pi} & e^{\ii\pi k} & e^{\ii\pi} \\
1 & e^{\ii\pi(k+1)/k} & e^{\ii\pi(k+1)} & e^{\ii\pi/k}
\end{bmatrix}
\end{equation}
Let $A = X\Lambda X^{-1}$, where $\Lambda = \diag(0,1,k,k+1)$.
be the adjacency matrix of $G_{k}$.
By Theorem \ref{thm:noncirculant-upst-graph}, $G_{k}$ has universal perfect state transfer. 

Note $G_{6}$ is a {non-circulant} graph with universal perfect state transfer.
The eigenvalues of $G_{6}$ are $\{0,1,6,7\}$ and its adjacency matrix is
\begin{equation}
A =
\begin{bmatrix}
0 & \thf(1+e^{-\ii\pi/6}) & \ohf & \thf(1-e^{-\ii\pi/6}) \\
\thf(1+e^{\ii\pi/6}) & 0 & \thf(1-e^{\ii\pi/6}) & \ohf \\
\ohf & \thf(1-e^{-\ii\pi/6}) & 0 & \thf(1+e^{-\ii\pi/6}) \\
\thf(1-e^{\ii\pi/6}) & \ohf & \thf(1+e^{\ii\pi/6}) & 0
\end{bmatrix}.
\end{equation}
\end{example}


\section{Denseness}

In this section, we show that if a circulant has universal perfect state transfer, 
then all of its coefficients must be nonzero under certain conditions on the order.

\begin{definition}
A circulant $\Circ(a_{0},a_{1},\ldots,a_{n-1})$ is called {\em dense} 
if $a_{j} \neq 0$ for $j=1,\ldots,n-1$.
\end{definition}

Our main result in this section is the following.

\begin{theorem} \label{thm:dense-conjecture}
Let $G = \Circ(a_{0},\ldots,a_{n-1})$ be a circulant with universal perfect state transfer.
If $n$ is a prime, square of a prime, or a power of two, then $G$ is dense.
\end{theorem}

\par\noindent
We will divide the proof of the theorem into several lemmas.
First, we consider the case when the order of the circulant is prime.

\begin{lemma} \label{lemma:dense-prime}
For a prime $p$, let $\Circ(a_{0},\ldots,a_{p-1})$ be a circulant 
with universal perfect state transfer.
Then, $a_{j} \neq 0$ for $j=1,\ldots,p-1$.

\begin{proof}
We may assume that $\lambda_{0} = 0$ (by a diagonal shift).
Since $a_{j} = \frac{1}{p} \sum_{k=1}^{p-1} \lambda_{k}\zeta_{p}^{-jk}$,
we have that $a_{j} \in \QQ(\zeta_{p})$.
Note that $\{\zeta_{p}^{j} : j=1,\ldots,p-1\}$ is a basis for the cyclotomic field extension
$\QQ(\zeta_{p})/\QQ$.
Thus, if $a_{j} = 0$ for some $j \neq 0$, then $\lambda_{k} = 0$ for all $k \neq 0$.
But, this is a contradiction to the assumed form of $\lambda_{k}$ in \eqref{eqn:simpler-form}.
\end{proof}
\end{lemma}

Second, we consider universal perfect state transfer in circulants
whose order is the square of a prime.

\begin{lemma} \label{lemma:gcd-map}
Given a prime $p$, let $n = p^{2}$.
For a circulant $\Circ(a_{0},\ldots,a_{n-1})$, suppose $a_{j} \in \QQ$ for some $j \neq 0$.
Then, $a_{d} = a_{j}$ for $d = \mygcd(j,n)$.

\begin{proof}
Let $j \in \{1,\ldots,n-1\}$ and $d = \mygcd(j,n)$.
Then, there is $\ell \in (\ZZ/n\ZZ)^{\star}$ so that $j \ell \equiv d \pmod{n}$.
To see this, note $j\ell \equiv d \pmod{n}$ is solvable since $\mygcd(j,n)$ divides $d$.
Moreover, $\mygcd(\ell,n) = 1$ since $\mygcd(\ell, n/d) = 1$ and $\mygcd(\ell,d) = 1$.
Here, we used the fact that $n = p^{2}$.

Since $\Gal(\QQ(\zeta_{n})/\QQ) = (\ZZ/n\ZZ)^{\star}$, there is a field automorphism
$\phi_{\ell}$ of $\QQ(\zeta_{n})$ which fixes $\QQ$ for which 
$\phi_{\ell}(\zeta_{n}) = \zeta_{n}^{\ell}$.
Since $a_{j} = (1/n)\sum_{k=0}^{n-1} \lambda_{k}\zeta_{n}^{-jk}$, we have
\begin{equation}
\phi_{\ell}(a_{j}) 
	= \frac{1}{n}\sum_{k=0}^{n-1} \lambda_{k} \zeta_{n}^{-\ell jk} 
	= \frac{1}{n}\sum_{k=0}^{n-1} \lambda_{k} \zeta_{n}^{-dk} 
	= a_{d}.
\end{equation}
But, $\phi_{\ell}(a_{j}) = a_{j}$ since $a_{j} \in \QQ$. This shows $a_{j} = a_{d}$.
\end{proof}
\end{lemma}

\begin{lemma} \label{lemma:odd-divisor}
Let $G = \Circ(a_{0},\ldots,a_{n-1})$ be a circulant of order $n$ with
universal perfect state transfer.
If $d$ is an odd divisor of $n$ and $n/d$ is prime, then $a_{d} \neq 0$.

\begin{proof}
Let $p = n/d$ be prime. Then,
\begin{equation} \label{eqn:coeff-at-divisor}
a_{d} = \frac{1}{n} \sum_{k=0}^{n-1} \lambda_{k}\zeta_{n}^{-dk}
	= \frac{1}{p} \sum_{k=0}^{p-1} \Lambda_{k}\zeta_{p}^{-k},
\end{equation}
where $\Lambda_{k} = \frac{1}{d}\sum_{\ell=0}^{d-1} \lambda_{k + \ell p}$.
By Corollary \ref{cor:simpler-eigenvalue-form},
we have 
\begin{equation}
\Lambda_{k} 
	= \frac{1}{d} \sum_{\ell=0}^{d-1} (k + \ell p + n\ZZ) 
	= k + \frac{p(d-1)}{2} + n\ZZ.
\end{equation}
Using this in \eqref{eqn:coeff-at-divisor} combined with the fact that
 $\sum_{k=0}^{p-1} \zeta_{p}^{-k} = 0$, we get
\begin{eqnarray}
a_{d} 
	& = & \frac{1}{p} \sum_{k=0}^{p-1} \left(k + \frac{p(d-1)}{2} + n\ZZ\right) \zeta_{p}^{-k} \\
	& = & \frac{1}{p} \sum_{k=1}^{p-1} k \zeta_{p}^{-k}.
\end{eqnarray}
Since $\zeta_{p}^{k}$, for $k=1,\ldots,p-1$, are linearly independent, 
we have $a_{d} \neq 0$.
\end{proof}
\end{lemma}

For our next lemma, we will need the following fact about connectivity in circulants.

\begin{fact} \label{fact:circulant-connected} (Meijer \cite{m91}, Theorem 4.2) \\
A circulant $\Circ(a_{0},\ldots,a_{n-1})$ is connected if and only if
$\mygcd(\{j : a_{j} \neq 0\} \cup \{n\}) = 1$.
\end{fact}

\begin{lemma} \label{lemma:dense-prime-squared}
For a prime $p$, suppose $n = p^{2}$.
Let $G = \Circ(a_{0},\ldots,a_{n-1})$ be a circulant with universal perfect state transfer.
Then, $a_{j} \neq 0$ for all $j \neq 0$.

\begin{proof}
Suppose $a_{j} = 0$ for some $j \neq 0$.
By Lemma \ref{lemma:gcd-map}, we have $a_{d} = 0$ for $d = \mygcd(j,n)$.
Since $n = p^{2}$, we have two cases to consider: $d = 1$ or $d = p$.
If $a_{1}=0$, then $G$ is not connected by Fact \ref{fact:circulant-connected}.
If $a_{p}=0$, then this contradicts Lemma \ref{lemma:odd-divisor}.
\end{proof}
\end{lemma}

Finally, we consider universal perfect state transfer in circulants
whose order is a power of two. 
Here, we use the following result of Good in a crucial manner.

\begin{fact} \label{fact:good} (Good \cite{g86}, Theorem 1) \\
If $m$ is a power of two, then $\{e^{\ii\pi r/m} : r=0,\ldots,m-1\}$
is linearly independent over $\QQ$.
\end{fact}

\begin{lemma} \label{lemma:dense-binary-powers}
For a positive integer $d$, suppose $n = 2^{d}$.
Let $G = \Circ(a_{0},\ldots,a_{n-1})$ be a circulant with universal perfect state transfer.
Then, $a_{j} \neq 0$ for all $j=1,\ldots,n-1$.

\begin{proof}
Given $d \in \ZZ^{+}$, let $n = 2^{d}$ and $m = 2^{d-1}$. 
Recall that $\zeta_{n} = e^{2\pi\ii/n}$.
We have
\begin{equation}
a_{j} = \frac{1}{n} \sum_{k=0}^{n-1} \lambda_{k}\zeta_{n}^{-jk}
	= \frac{1}{2m} \sum_{\ell=0}^{m-1} \Lambda_{\ell}(j) (e^{\ii\pi/m})^{-\ell},
\end{equation}
where
\begin{equation}
\Lambda_{\ell}(j) = 
	\sum_{k: jk \equiv \ell} \lambda_{k}
	-
	\sum_{k: jk \equiv m + \ell} \lambda_{k}
\end{equation}
By Fact \ref{fact:good}, if $a_{j} = 0$ then $\Lambda_{\ell}(j) = 0$ for all $\ell=0,\ldots,m-1$.
We show that $\Lambda_{0} \neq 0$.

We consider the case when $\ell = 0$.
If $j$ is odd, then the map $f_{j}(k) \equiv jk \pmod n$ is a bijection.
Thus, $\Lambda_{0}(j) = \lambda_{0} - \lambda_{m} \equiv m \pmod n$.
Next, suppose $j$ is even with $j = 2^{e}s$ where $e \ge 1$ and $s$ is odd.
Then, the values $k$ for which $jk \equiv 0 \pmod n$ are given by
$r 2^{d-e}$ for $r=0,1,\ldots,2^{e}-1$.
Also, the values $k$ for which $jk \equiv m \pmod n$ are given by
$(2r+1) 2^{d-e-1}$ for $r=0,1,\ldots,2^{e}-1$.
Therefore,
\begin{equation}
\Lambda_{0} 
	= \sum_{r=0}^{2^{e}-1} (r 2^{d-e} - (2r+1) 2^{d-e-1})
	\equiv m \pmod n.
\end{equation}
Thus, in both cases we have $\Lambda_{0} \not\equiv 0 \pmod n$, which implies $\Lambda_{0} \neq 0$.
\end{proof}
\end{lemma}

\begin{proof} (of Theorem \ref{thm:dense-conjecture}) \\
Follows immediately from 
Lemmas \ref{lemma:dense-prime}, \ref{lemma:dense-prime-squared}, and \ref{lemma:dense-binary-powers}.
\end{proof}

\subsection{Non-dense circulants with universal state transfer}

We show a partial converse of Theorem \ref{thm:dense-conjecture} by constructing non-dense 
circulants with universal perfect state transfer whose orders are not prime powers. 
This observation uses the following fact about cyclotomic units.

\begin{fact} \label{fact:washington}
(Washington \cite{w97}, Proposition 2.8) \\
Suppose $n$ is a positive integer which has at least two distinct prime factors. 
Then $1 - \zeta_{n}$ is a unit of $\ZZ[\zeta_{n}]$.
Moreover, 
\begin{equation}
\prod_{\substack{0 < j < n\\\mygcd(j,n)=1}} (1 - \zeta_{n}^{j}) = 1.
\end{equation}
\end{fact}

\par\noindent
Note that Fact \ref{fact:washington} also implies that $1 - \zeta_{n}^{j}$ is a unit of
$\ZZ[\zeta_{n}]$ for every $j$ with $\mygcd(j,n) = 1$.

\bigskip

\begin{proposition} \label{prop:nondense-circ-upst}
For two distinct primes $p$ and $q$, let $n = pq$. 
Let $c_{0},\ldots,c_{n-1}$ be integers so 
\begin{equation} \label{eqn:a1-is-zero}
\sum_{k=0}^{n-1} c_{k}\zeta_{n}^{-k} = \frac{1}{1 - \zeta_{n}^{-1}}.
\end{equation}
Let $a_{0}=0$ 
and, for $j=1,\ldots,n-1$, let
$a_{j} = 1/(\zeta_{n}^{-j} - 1) + \sum_{k=0}^{n-1} c_{k}\zeta_{n}^{-jk}$.
Then, $G = \Circ(a_{0},\ldots,a_{n-1})$ is a non-dense circulant with universal perfect state transfer.

\begin{proof}
By the choice of the integers $c_{k}$ in \eqref{eqn:a1-is-zero}, $a_{1} = a_{n-1} = 0$. This shows
that $G$ is not dense. Also, note that
\begin{equation} \label{eqn:scalar-connection}
\lambda_{0} = \sum_{j=0}^{n-1} a_{j} = \sum_{j=1}^{n-1} \frac{1}{\zeta_{n}^{-j}-1} - \sum_{k=0}^{n-1} c_{k}.
\end{equation}
The other eigenvalues are given by $\lambda_{\ell} = \sum_{j=1}^{n-1} a_{j}\zeta_{n}^{j\ell}$, for $\ell \neq 0$.
By definition of $a_{j}$, 
\begin{equation}
\lambda_{\ell} 
	= \sum_{j=1}^{n-1} \left(\frac{1}{\zeta_{n}^{-j}-1} + \sum_{k=0}^{n-1} c_{k}\zeta_{n}^{-jk}\right) \zeta_{n}^{j\ell}
	= \sum_{j=1}^{n-1} \frac{\zeta_{n}^{j\ell}}{\zeta_{n}^{-j}-1} + c_{\ell}n - \sum_{k=0}^{n-1} c_{k}.
\end{equation}
Using \eqref{eqn:scalar-connection}, we get
\begin{equation}
\lambda_{\ell} 
	= \sum_{j=1}^{n-1} \frac{\zeta_{n}^{j\ell}-1}{\zeta_{n}^{-j}-1} + c_{\ell}n + \lambda_{0} \\ 
	= (n-1) - \sum_{j=1}^{n-1} \frac{1-\zeta_{n}^{j(\ell+1)}}{1-\zeta_{n}^{j}} + c_{\ell}n + \lambda_{0}.
\end{equation}
But, note that
\begin{equation}
\sum_{j=1}^{n-1} \frac{1-\zeta_{n}^{j(\ell+1)}}{1-\zeta_{n}^{j}}
	= \sum_{j=1}^{n-1} \sum_{s=0}^{\ell} \zeta_{n}^{sj} = n-1-\ell.
\end{equation}
This shows that $\lambda_{\ell} = \ell + c_{\ell}n + \lambda_{0}$, for $\ell=1,\ldots,n-1$.
By Theorem \ref{thm:circulant-upst-eigenvalue}, this shows $G$ has universal perfect state transfer.

Here, we confirm that the underlying graph of $G$ is connected.
Since $1-\zeta_{p}^{-1}$ is not a unit of $\ZZ[\zeta_{p}]$, for any prime $p$,
we have
\begin{equation}
a_{q} = \frac{1}{\zeta_{n}^{-q}-1} + \sum_{k=0}^{n} c_{k}\zeta_{n}^{-qk}
	= \frac{1}{\zeta_{p}^{-1}-1} + \sum_{k=0}^{n-1} c_{k}\zeta_{p}^{-k}
	\neq 0.
\end{equation}
since $c_{k}$ are all integers.
Similarly, $a_{p} \neq 0$.
Since $\mygcd(p,q)=1$, Fact \ref{fact:circulant-connected} implies $G$ is connected.
\end{proof}
\end{proposition}

\bigskip

\newcommand{\fhf}{\tfrac{5}{2}}
\newcommand{\isqt}{\tfrac{\ii}{\sqrt{3}}}

\begin{example}
We show a circulant $G$ of order $n=6$ with universal perfect state transfer which is {\em not} dense. 
Here, we have 
$\zeta_{6} = \tfrac{1}{2} + \ii\tfrac{\sqrt{3}}{2}$,
$\zeta_{6}^{2} = -\tfrac{1}{2} + \ii\tfrac{\sqrt{3}}{2}$,
$\zeta_{6}^{3} = -1$,
$\zeta_{6}^{4} = -\tfrac{1}{2} - \ii\tfrac{\sqrt{3}}{2}$,
$\zeta_{6}^{5} = \tfrac{1}{2} - \ii\tfrac{\sqrt{3}}{2}$.
Note that
\begin{equation}
a_{1} = \frac{1}{\zeta_{6}^{-1} - 1} + \sum_{k=0}^{5} c_{k}\zeta_{6}^{-k}
	= \frac{1}{\zeta_{6}^{5} - 1} + (1 - \zeta_{6}) = 0.
\end{equation}
So, in Theorem \ref{thm:circ-coeff-upst}, we choose
$c_{0} = 1$, $c_{1} = c_{2} = c_{3} = c_{4} = 0$ and $c_{5} = -1$.
The other coefficients can be computed using
$a_{j} = 1/(\zeta_{6}^{-1}-1) + \sum_{k=0}^{5} c_{k}\zeta_{6}^{-jk}$.
By a straightforward computation,
$a_{0} = \tfrac{5}{2}$, $a_{2} = 1 - \isqt$, $a_{3} = \thf$, $a_{4} = 1 + \isqt$ 
and, of course, $a_{5} = 0$.

Hence, the adjacency matrix of $G$ is given by
\begin{equation}
\begin{bmatrix}
\fhf & 0 & 1-\isqt & \thf & 1+\isqt & 0 \\
0 & \fhf & 0 & 1-\isqt & \thf & 1+\isqt \\
1+\isqt & 0 & \fhf & 0 & 1-\isqt & \thf \\
\thf & 1+\isqt & 0 & \fhf & 0 & 1-\isqt \\
1-\isqt & \thf & 1+\isqt & 0 & \fhf & 0 \\
0 & 1-\isqt & \thf & 1+\isqt & 0 & \fhf
\end{bmatrix}.
\end{equation}
The eigenvalues of $G$ are given by $\lambda_{k} = \sum_{j=0}^{5} a_{j}\zeta_{6}^{jk}$.
We confirm the eigenvalue form in \eqref{eqn:simpler-form} that
$\lambda_{\ell} = \ell + 6c_{\ell}$, for $\ell=0,1,\ldots,5$.
It can be verified that
$\lambda_{0} = 6 = 0 + 6c_{0}$, 
$\lambda_{1} = 1 = 1 + 6c_{1}$, 
$\lambda_{2} = 2 = 2 + 6c_{2}$, 
$\lambda_{3} = 3 = 3 + 6c_{3}$, 
$\lambda_{4} = 4 = 4 + 6c_{4}$, 
and $\lambda_{5} = -1 = 5 + 6c_{5}$.
\end{example}


\section{Property $\TT$}

In this concluding section, we consider universal perfect state transfer in complex unit gain graphs.
First, we observe the following fact.

\begin{fact}
$\Circ(0,-\ii,\ii)$ is the only graph on $3$ vertices with universal perfect state transfer,
up to switching equivalence.

\begin{proof}
Cameron \etal \cite{cfghst14} showed that $\Circ(0,-\ii,\ii)$ has universal perfect state transfer.
Hence,  it suffices to show that any graph on $3$ vertices with universal perfect state transfer 
must be a circulant. 

Let $G$ be a graph with Hermitian adjacency matrix $A$.
Suppose $M$ is a type-II matrix of the form
\begin{equation}
M = 
\begin{bmatrix}
1 & 1 & 1 \\
1 & e^{\ii\alpha} & e^{\ii\beta} \\
1 & e^{\ii\beta} & e^{\ii\alpha}
\end{bmatrix}
\end{equation}
which diagonalizes $A$.
By Corollary \ref{cor:flatness}, we have that $1 + e^{\ii\alpha} + e^{\ii\beta} = 0$.
This shows $\cos(\alpha) + \cos(\beta) = -1$ and $\sin(\alpha) + \sin(\beta) = 0$.
This implies $\alpha = 2\pi/3$ and $\beta = 2\alpha$.
Thus, $M/\sqrt{3}$ is the Fourier matrix which diagonalizes any circulant 
matrix of order $3$.
\end{proof}
\end{fact}

Godsil \cite{g12b} proved that for any constant $k$, there is only a finite number of
(unweighted) graphs with maximum degree $k$ with perfect state transfer. This shows that perfect
state transfer is a rare phenomenon (in the absence of weights).
This motivates our next conjecture.

We say a graph has {\em property $\TT$} if all of the nonzero coefficients in its adjacency matrix 
are complex numbers with unit magnitude. 

\begin{conjecture}
$\Circ(0,-\ii,\ii)$ is the only circulant with property $\TT$ which has universal perfect state transfer.
\end{conjecture}


\section*{Acknowledgments}

Research supported by NSF grant DMS-1262737. 
Part of this work was started while C.T. was visiting Institut Henri Poincar\'{e} (Centre \'{E}mile Borel)
and University of Waterloo. This author would like to thank IHP and Chris Godsil for hospitality 
and support.


\end{document}